\documentclass{llncs}
\usepackage{makeidx}  
\usepackage[title]{appendix}

\usepackage{algorithm}
\usepackage{algorithmic}
\usepackage{amsmath}

\usepackage{listings}
\usepackage{syntax}
\usepackage{amssymb}
\usepackage{mathpartir}

\newcommand{\eg}{\emph{e.g.}}
\newcommand{\ie}{\emph{i.e.}}

\newcommand{\code}[1]{\texttt{#1}}

\newcommand{\expr}{\code{e}}

\newcommand{\procname}{p}
\newcommand{\proc}{\textsc{\procname}}
\newcommand{\procP}{\textsc{\procname}}

\newcommand{\stmt}{\code{S}}

\newcommand{\stmtSA}{{\ensuremath{\code{S}_1}}}
\newcommand{\stmtSB}{{\ensuremath{\code{S}_2}}}
\newcommand{\stmtSC}{{\ensuremath{\code{S}_3}}}

\newcommand{\logicalformula}[1]{\varphi^{#1}}
\newcommand{\inv}{\logicalformula{inv}}
\newcommand{\post}{\logicalformula{post}}
\newcommand{\path}{\logicalformula{path}}

\newcommand{\pre}{\logicalformula{pre}}
\newcommand{\vc}{\logicalformula{vc}}

\newcommand{\mainvcfn}{\textsc{postvc}}
\newcommand{\mainvc}[2]{\mainvcfn({#1},{#2})}

\newcommand{\cspre}{\pre_{cs}}
\newcommand{\cspost}{\post_{cs}}

\newcommand{\simplifyfn}{\textsc{Simplify}}
\newcommand{\nextfn}{\textsc{Next}}

\newcommand{\postfn}{\textsc{post}}
\newcommand{\vcfn}{\textsc{vc}}
\newcommand{\tbody}{\textsc{TB}}
\newcommand{\initstatefn}{\textsc{init}}

\newcommand{\pureinv}{\mathrm{pure}(\inv)}

\newcommand{\vals}{\mathcal{V}}
\newcommand{\gvars}{G}
\newcommand{\lvars}{L}
\newcommand{\gmap}{\rho_g}
\newcommand{\gmaps}{\Sigma_G}
\newcommand{\lmap}{\rho_\ell}
\newcommand{\lmaps}{\Sigma_L}
\newcommand{\cstates}{\Sigma_C}

\newcommand{\lstack}{\gamma}
\newcommand{\initstate}{\sigma_{\text{init}}}

\newcommand{\initial}[1]{\textit{initial}(#1)}
\newcommand{\final}[1]{\textit{final}(#1)}
\newcommand{\inputval}[1]{\textit{input}(#1)}
\newcommand{\outputval}[1]{\textit{output}(#1)}

\newcommand{\iosem}[1]{\rightsquigarrow_{#1}}
\newcommand{\iosemP}{\iosem{\proc}}
\newcommand{\sssem}[1]{\rightarrow_{#1}}
\newcommand{\sssemP}{\sssem{\proc}}

\newcommand{\basicevalsto}[3]{ ({#1},{#2}) \Downarrow {#3}}
\newcommand{\evalsto}[4]{ ({#1} \uplus {#2},{#3}) \Downarrow {#4}}

\newcommand{\labrule}[3]{\inferrule*[Lab={[#1]}]{#2}{#3}}

\mathchardef \mhyphen="2D

\lstdefinestyle{mystyle}{
  basicstyle=\footnotesize, 
  captionpos=b,
  numbers=left,
  numbersep=6pt,                  
  morekeywords={if, for, else,
  procedure,modifies, var, returns, int, call, assume},
}
\begin{document}
\frontmatter          
\pagestyle{headings}  
\mainmatter              
\title{Checking Observational Purity of Procedures}
\titlerunning{Checking Observational Purity}  
%

\author{Himanshu Arora\inst{1} \and
Raghavan Komondoor\inst{1} \and
G. Ramalingam\inst{2}
}
%
%
 \institute{Indian Institute of Science, Bangalore\\
  \email{\{himanshua, raghavan\}@iisc.ac.in},
 \and
 Microsoft Research \\
 \email{grama@microsoft.com}}

\maketitle              

\pagestyle{empty}
\begin{abstract}
  Verifying whether a procedure is \emph{observationally pure} is useful in
  many software engineering scenarios. An observationally pure procedure
  always returns the same value for the same argument, and thus mimics a
  mathematical function. The problem is challenging when procedures use
  private mutable global variables, e.g., for memoization of frequently
  returned answers, and when they involve recursion.

  We present a novel verification approach for this problem. Our approach
  involves encoding the procedure's code as a formula that is a disjunction
  of path constraints, with the recursive calls being replaced in the
  formula with references to a mathematical function symbol. Then, a
  theorem prover is invoked to check whether the formula that has been
  constructed agrees with the function symbol referred to above in terms of
  input-output behavior for all arguments.


  We evaluate our approach on a set of realistic examples, using the Boogie
  intermediate language and theorem prover. Our evaluation shows that the
  invariants are easy to construct manually, and that our approach is
  effective at verifying observationally pure procedures.
\end{abstract}

\section{Introduction}

A procedure in an imperative programming language is said to be
\emph{observationally pure} (OP) if for each specific argument value it has
a specific return value, across all possible sequences of calls to the
procedure, irrespective of what other code runs between these calls.  In
other words, the input-output behavior of an OP procedure mimics a mathematical
function.


Any procedure whose code is deterministic and does not read any
pre-existing state other than its arguments is trivially OP.
However, it is common for procedures, especially ones in libraries,
to update and read global variables, typically to optimize their own behavior,
while still mimicking mathematical functions in terms of their input-output behavior.
In this paper, we focus on the problem of checking observational purity of
procedures that read and write global variables, especially in the presence of recursion,
which makes the problem harder.

\subsection{Motivating Example}
We use the example procedure `factCache' in
Listing~\ref{lst:factorialSimple}  as our running example. It
returns n! for the given argument n, and caches the return value for the
last argument provided to it. It uses two ``private'' global variables to
implement the caching -- \code{g}, and \code{lastN}. \code{g} is initialized to -1, and after
the first call to the procedure onwards is used to store the return value
from the most recent call. \code{lastN} is used to store the argument received in
the most recent call. Clearly this procedure is OP, and mimics the
input-output behavior of a regular factorial procedure that does not cache
any results. 

\begin{lstlisting}[float,language=c,basicstyle=\scriptsize,caption= {Procedure factCache:
      returns n!, and memoizes most recent result.},
    label=lst:factorialSimple]
  
int g := -1;
int lastN := 0;
int factCache( int n) {
  if(n <= 1) {
    result := 1;
  } else if (g != -1 && n == lastN) {
    result := g;
  } else {
    g = n * factCache( n - 1 );
    lastN = n;
    result := g;
  }
  return result;
}
\end{lstlisting}

\subsection{Proposed Approach}
\label{ssec:intro:approach}
%
%
%
Our verification approach is based on Floyd-Hoare logic. In order to verify
a recursive procedure inductively, typically a specification of the
procedure would need to be provided. The first idea for such a
specification would be a full functional specification of the
procedure. That is, if someone specifies that factCache mimics n!, then the
verifier could replace Line~10 in the code with `g = n * (n-1)!'. This, on
paper, is sufficient to assert that Line 12 always assigns n! to
`result'. However, to establish that Line~8 also does the same, an invariant
would need to be provided that describes the possible values of g before
any invocation to the procedure. In our example, a suitable invariant would
be `(g = -1) $\vee$ (g = lastN!)'. The verifier would also need to verify
that at the procedure's exit the invariant is re-established. Lines 10-12,
with the recursive call replaced by (n-1)!, suffices on paper to
re-establish the invariant.

The candidate approach mentioned above, while being plausible, suffers from
two weaknesses. The first is that a human would need to guess the
mathematical expression that is implemented by the given procedure. This may not be
easy when the procedure is complex. Second, the underlying theorem prover
would need to prove complex arithmetic properties, e.g., that n * (n-1)! is
equal to n!. Complex proofs such as this may be out of bounds for many
existing theorem provers.

Our key insight is to sidestep the challenges mentioned by introducing a
function symbol, say \emph{factCache}, and replacing the recursive call for the purposes
of verification with this function symbol. Intuitively, \emph{factCache} represents the
mathematical function that the given procedure mimics if the procedure is
OP.  In our example, Line~10 would become `g = n * \emph{factCache}(n-1)'. This step
needs no human involvement. The approach needs an invariant; however, in a
novel manner, we allow the invariant also to refer to \emph{factCache}. In our example,
a suitable invariant would be `(g = -1) $\vee$ (g = lastN *
\emph{factCache}(lastN-1))'. This sort of invariant is relatively easy to construct;
e.g., a human could arrive at it just by looking at Line~2 and with a local
reasoning on Lines~10 and~11. Given this invariant, (a) a theorem prover
could infer that the condition in Line~7 implies that Line~8 necessarily
copies the value of `n * \emph{factCache}(n-1)' into `result'. Due to the transformation to Line~10 mentioned above,
(b) the theorem prover can infer that Line~12 also does the same. Note that since these two expressions are syntactically
identical, a theorem prover can easily establish that they are equal in
value.  Finally, since Line~6 is reached under a different condition than
Lines~8 and~12, the verifier has finished establishing that the procedure
always returns the same expression in n for any given value of n.

Similarly, using the modified Line~10 mentioned above and from Line~11, the
prover can re-establish that g is equal to `lastN * \emph{factCache}(lastN - 1)' when
control reaches Line~12. Hence, the necessary step of proving the given
invariant to be a valid invariant is also complete. 

Note, the effectiveness of the approach depends on the nature of the given
invariant. For instance, if the given invariant was `(g = -1) $\vee$ (g =
lastN!)', which is also technically correct, then the theorem prover may
not be able to establish that in Lines~8 and~12 the variable `g' always
stores the same expression in n. However, it is our claim that in fact it is the
invariant `(g = -1) $\vee$ (g = lastN * \emph{factCache}(lastN-1))' that is
easier to infer by a human or by a potential tool, as justified by us two
paragraphs above.

\subsection{Salient Aspects Of Our Approach}

This paper makes two significant contributions. First, it tackles the
circularity problem that arises due to the use of a presumed-to-be OP
procedure in assertions and invariants and the use of these invariants in
proving the procedure to be OP. This requires us to prove the soundness of
an approach that verifies observational purity as well the validity of
invariants simultaneously (as they cannot be decoupled).

Secondly, as we show, a direct approach to this verification problem (which we
call the existential approach) reduces it to a problem of verifying that a logical formula
is a tautology. The structure of the generated formula, however, makes the resulting
theorem prover instances hard. We show how a conservative approximation can be
used to convert this hard problem into an easier problem of checking satisfiability
of a quantifier-free formula, which is something within the scope of state-of-the-art theorem
provers.


The most closely related previous approaches are by Barnett et
al.~\cite{barnett200499,barnett2006allowing}, and by
Naumann~\cite{naumann2007observational}.  These approaches check observational
purity of procedures that maintain mutable global state. However, none of
these approaches use a function symbol in place of recursive calls or
within invariants. Therefore, it is not clear that these approaches can
verify recursive procedures. 
Barnett et al., in fact, state ``there is a circularity - it would take a delicate argument, and additional conditions,
to avoid unsoundness in this case''.
To the best of our knowledge ours is the first paper to show that it is
feasible to check observational purity of procedures that maintain mutable
global state for optimization purposes and that make use of recursion.

Being able to verify that a procedure is OP has many potential
applications. The most obvious one is that OP procedures can be
memoized. That is, input-output pairs can be recorded in a table, and calls
to the procedure can be elided whenever an argument is seen more than
once. This would not change the semantics of the overall program that calls
the procedure, because the procedure always returns the same value for the
same argument (and mutates only private global variables). Another
application is that if a loop contains a call to an OP procedure, then the
loop can be parallelized (provided the procedure is modified to access and
update its private global variables in a single atomic operation).


The rest of this paper is structured as
follows. Section~\ref{sec:background} introduces the core programming
language that we address. Section~\ref{sec:semantics} provides formal
semantics for our language, as well as definitions of invariants and
observational purity. Section~\ref{sec:vcgen} describes our approach
formally. Section~\ref{sec:invariant} discusses an approach for generating
an invariant automatically in certain cases. Section~\ref{sec:experiments}  describes evaluation of our approach on a few realistic examples.
Section~\ref{sec:related} describes related work.

\newcommand{\elt}{\ensuremath{\in} }
\newcommand{\domain}[1]{#1}

\section{Language Syntax}
\label{sec:background}

In this paper, we assume that the input to the purity checker is a library consisting
of one or more procedures, with shared state consisting of one or more variables
that are private to the library. We refer to these variables as ``global'' variables to
indicate that they retain their values across multiple invocations of the library
procedures, but they cannot be accessed or modified by procedures outside
the library (that is, the clients of the library).

\begin{figure}[t!]
{\tt
\begin{tabular}{rll}
L \elt & \domain{Lib} & ::= $\overline{\code{g := c}}$ $\overline{\code{P}}$ \\
P \elt & \domain{Proc} & ::= p (x) \{ S; return y \} \\
S \elt & \domain{Stmt} & ::=  x := e | x := p(y) | S ; S | if (e) then S else S \\
e \elt & \domain{Expr} & ::= c | x | e op e | unop e \\
op \elt & \domain{Ops} & ::= + | - | / | * | \% | > | < | == | $\wedge$ | $\vee$ \\
unop \elt & \domain{UnOps} & ::= $\neg$ \\
\multicolumn{3}{c}{
x, y \elt  \domain{LocalId} $\cup$ \domain{GlobalId},
g \elt \domain{GlobalId},
c \elt $\vals$,
p \elt \domain{ProcId}
}
\end{tabular}
}
\caption{Programming language syntax and meta-variables}
 \label{fig:grammar}
\end{figure}

In Fig.~\ref{fig:grammar}, we present the syntax of a simple
programming language that we address in this paper.
Given the foundational focus of this work, we keep the programming
language very simple, but the ideas we present can be generalized.
A \code{return} statement is required in each procedure,
and is permitted only as the last
statement of the procedure.
The language does not contain any looping construct.
Loops can be modelled as recursive procedures.
The formal parameters of a procedure are readonly and cannot be
modified within the procedure.
We omit types from the language. We permit only variables of primitive types.
In particular, the language does not allow pointers or dynamic memory allocation.
Note that expressions are pure in this language, and a procedure call
is not allowed in an expression. Each procedure call is modelled as a
separate statement.

For simplicity of presentation, without loss of conceptual generality, we assume
that the library consists of a single (possibly recursive) procedure, with a single formal
parameter.
In the sequel, we will use the symbol $\proc$ (as a metavariable) to
represent this library procedure, $\procname$ (as a metavariable) to
represent the \emph{name} of this procedure, and will assume that the name
of the formal parameter is \code{n}. 
If the procedure is of the form ``$\procname$ \code{ (n) \{ S; return r \}}'', we refer to \code{r} as the \emph{return}
variable, and  refer to ``\code{S; return r}'' as the \emph{procedure body}
and denote it as $\text{body}(\proc)$.
The library also contains, outside of the procedure's code,
a sequence of initializing declarations of
the global variables used in the procedure, of the form ``\code{g1 := c1};
$\ldots$; \code{gN := cN}''. These initializations are assumed to be
performed  once during any execution of the client application,
just before the first call to the
procedure $\proc$ is placed by the client application.

 Finally, a note about terminology: throughout this paper we
 use the word `procedure' to refer to the library procedure $\proc$, and
 use the word
`function' to refer to a mathematical function.

\section{A Semantic Definition of Purity}
\label{sec:semantics}

In this section, we formalize the input-output semantics of the procedure $\proc$ as a relation $\iosem{\proc}$,
where $n \iosem{\proc} r$ indicates that an invocation of $\proc$ with input $n$ may return a result of $r$.
The procedure is then defined to be observationally pure if the relation $\iosem{\proc}$ is a (partial) function:
that is, if  $n \iosem{\proc} r_1$ and $n \iosem{\proc} r_2$, then $r_1 = r_2$.

The object of our analysis is the set of procedures in the library, or,
actually, a single-procedure library (for simplicity of presentation), but
not the entire (client) application. The result of our analysis is valid
for any client program that uses the procedure/library.  The only
assumptions we make are: (a) The shared state used by the library (the
global variables) are private to the library and cannot be modified by the
rest of the program, and (b) The client invokes the library procedures
sequentially: no concurrent or overlapping invocations of the library
procedures by a concurrent client are permitted.

The following semantic formalism is motivated by the above observations. It can be seen as the semantics
of the so-called ``most general sequential client'' of procedure $\proc$, which is the program:
\code{while (*) { x = $\procname$ (random()); }}.
The executions (of $\proc$) produced by this program include all possible executions (of $\proc$)  produced by all
sequential clients.

Let $\gvars$ denote the set of global variables. Let $\lvars$ denote the set of local variables.
Let $\vals$ denote the set of numeric values (that the variables can take).
An element $\gmap \in \gmaps = \gvars \hookrightarrow \vals$ maps global variables to their values.
An element $\lmap \in \lmaps = \lvars \hookrightarrow \vals$ maps local variables to their values.
We define a \emph{local continuation} to be a statement sequence ending with a \code{return} statement.
We use a local continuation to represent the part of the procedure body that still remains to be
executed. Let $\cstates$ represent the set of local continuations.
The set of runtime states (or simply, \emph{states}) is defined to be
$(\cstates \times \lmaps)^* \times \gmaps$, where the first component
represents a runtime stack, and the second component the values of global
variables. We denote individual states using symbols $\sigma, \sigma_1,
\sigma_i$, etc. The runtime stack is a sequence, each element of which is a pair
$(\stmt,\lmap)$ consisting of the remaining procedure fragment $\stmt$ to
be executed and the values of local variables $\lmap$.  We write
$(\stmt,\lmap)\gamma$ to indicate a stack where the topmost entry is
$(\stmt,\lmap)$ and $\gamma$ represents the remaining part of the stack.

We say that a state $((\stmt,\lmap)\gamma,\gmap)$ is an \emph{entry-state}
if its location is at 
the procedure entry point (\ie, if $\stmt$ is the entire body of the procedure),
and we say that it is an \emph{exit-state} if its location is at the procedure exit point
(\ie, if $\stmt$ consists of just a \code{return} statement).

A procedure $\proc$ determines a single-step execution relation
$\sssem{\proc}$, where $\sigma_1 \sssem{\proc} \sigma_2$ indicates that
execution proceeds from state $\sigma_1$ to state $\sigma_2$ in a single
step.  Fig.~\ref{fig:semantics} defines this semantics.  The semantics of
evaluation of a side-effect-free expression is captured by a relation
$\basicevalsto{\rho}{\expr}{v}$, indicating that the expression $\expr$
evaluates to value $v$ in an \emph{environment} $\rho$ (by
\emph{environment}, we mean an
element of $\lmaps \times \gmaps$).  We omit the definition of this
relation, which is straightforward.  We use the notation $\rho_1 \uplus
\rho_2$ to denote the union of two disjoint maps $\rho_1$ and $\rho_2$.

Note that most rules captures the usual semantics of the language constructs.
The last two rules, however, capture the semantics of the most-general sequential
client explained previously: when the call stack is empty, a new invocation of
the procedure may be initiated (with an arbitrary parameter value).


\begin{figure}
\begin{small}
\begin{mathpar}
\labrule{assgn}{
\code{x} \in \lvars \\
\evalsto{\lmap}{\gmap}{\code{e}}{v}
}{
((\code{x := e; S}, \lmap) \lstack, \gmap)
\sssemP 
((\code{S}, \lmap[\code{x} \mapsto v]) \lstack, \gmap)
}

\labrule{assgn}{
\code{x} \in \gvars \\
\evalsto{\lmap}{\gmap}{\code{e}}{v}
}{
((\code{x := e; S}, \lmap) \lstack, \gmap)
\sssemP 
((\code{S}, \lmap) \lstack, \gmap[\code{x} \mapsto v])
}

\labrule{seq}{}{
(((\stmtSA ; \stmtSB) ; \stmtSC , \lmap) \lstack, \gmap) 
\sssemP
((\stmtSA ; (\stmtSB ; \stmtSC) , \lmap) \lstack, \gmap) 
}

\labrule{if-true}{
\evalsto{\lmap}{\gmap}{\expr}{\code{true}}
}{
( (\code{(if (\expr) then \stmtSA else \stmtSB); \stmtSC}, \lmap) \lstack, \gmap)
\sssemP
( (\code{\stmtSA; \stmtSC}, \lmap) \lstack, \gmap)
}

\labrule{if-false}{
\evalsto{\lmap}{\gmap}{\expr}{\code{false}}
}{
( (\code{(if (\expr) then \stmtSA{} else \stmtSB); \stmtSC}, \lmap) \lstack, \gmap)
\sssemP
( (\code{\stmtSB; \stmtSC}, \lmap) \lstack, \gmap)
}

\labrule{call}{
\evalsto{\lmap}{\gmap}{\code{e}}{v} \\
\proc = \code{$\procname$(n) \stmtSA}
}{
((\code{y := $\procname$(e); \stmtSB}, \lmap) \lstack, \gmap)
\sssemP
((\stmtSA, [n \mapsto v]) (\code{y := $\procname$(e); \stmtSB}, \lmap) \lstack, \gmap)
}

\labrule{return}{
\evalsto{\lmap}{\gmap}{\code{r}}{v}
}{
((\code{return r}, \lmap) (\code{y := $\procname$(e); \stmt}, \lmap') \lstack, \gmap)
\sssemP
(\stmt, \lmap'[ \code{y} \mapsto v]) \lstack, \gmap)
}

\labrule{top-level-call}{
\code{B} = \text{body}(\proc) \\
v \in \vals
}{
([], \gmap)
\sssemP
([(\code{B}, [n \mapsto v])], \gmap)
}

\labrule{top-level-return}{
}{
[(\code{return r}, \lmap)], \gmap)
\sssemP
([], \gmap)
}

\end{mathpar}
\end{small}
\caption{
A small-step operational semantics for our language, represented as a relation $\sigma_1 \sssemP \sigma_2$.
Note that a state $\sigma_i$ is a configuration of the form
$((\stmt, \lmap) \lstack, \gmap)$ where
$\stmt$ captures the statements to be executed in the current procedure, 
$\lmap$ assigns values to local variables in the current procedure,
$\lstack$ is the call-stack (excluding the current procedure),
and $\gmap$ assigns values to global variables.
}
\label{fig:semantics}
\end{figure}

Note that all the following definitions are parametric over a given procedure $\proc$.
E.g., we will use the word ``execution'' as shorthand for ``execution of $\proc$''.

We define an \emph{execution} (of $\proc$) to be a sequence of states $\sigma_0 \sigma_1 \cdots \sigma_n$ such that
$\sigma_i \sssem{\proc} \sigma_{i+1}$ for all $0 \leq i < n$.
Let $\initstate$ denote the \emph{initial state} of the library; i.e., this
is the element of $\gmaps$ that is induced by the sequence of initializing declarations of
the library, namely, ``\code{g1 := c1};
$\ldots$; \code{gN := cN}''
We say that an execution $\sigma_0 \sigma_1 \cdots \sigma_n$ is a
\emph{feasible} execution if $\sigma_0 = \initstate$. Note, 
intuitively, a feasible execution  corresponds to the sequence
of states visited within the library across all invocations of
the library procedure over the course of a single execution of
the most-general client mentioned above; also, since the most-general client
supplies a random parameter value to each invocation of $\proc$, in general
multiple feasible executions of the library may exist.
We say that a state $\sigma$ is \emph{ reachable} if there exists an execution $\pi = \initstate \sigma_1 \cdots \sigma$.


We define a \emph{trace} (of $\proc$) to be a substring 
 $\pi = \sigma_0 \cdots \sigma_n$ of a feasible execution such that:
 (a) $\sigma_0$ is entry-state
 (b) $\sigma_n$ is an exit-state, and
 (c) $\sigma_n$ corresponds to the return from the invocation represented
by $\sigma_0$.
In other words, a trace is a state sequence corresponding to a single
invocation of the procedure. A trace may contain within it nested
sub-traces due to recursive calls, which are themselves traces.
Given a trace $\pi = \sigma_0 \cdots \sigma_n$, we define
$\initial{\pi}$ to be $\sigma_0$,
$\final{\pi}$ to be $\sigma_n$,
$\inputval{\pi}$ to be value of the input parameter in $\initial{\pi}$,
and $\outputval{\pi}$ to be the value of the return variable in $\final{\pi}$.


We define the relation $\iosemP$ to be $\{ (\inputval{\pi},\outputval{\pi}) \; | \; \pi \text{ is a trace of } \proc \}$.

\begin{definition}[Observational Purity]
\label{def:purity}
A procedure $\proc$ is said to be \emph{observationally pure} if the relation $\iosem{\proc}$ is a (partial) function:
that is, if for all $n$, $r_1$, $r_2$, if  $n \iosem{\proc} r_1$ and $n \iosem{\proc} r_2$, then $r_1 = r_2$.
\end{definition}

\subsubsection{Logical Formula and Invariants.}

Our methodology makes use of \emph{logical formulae} for different
purposes, including to express a given \emph{invariant}. 
Our logical formulae use the local and global variables in the library
procedure as free variables, use
the same operators as allowed in our language, and make use of universal as
well as existential quantification. 
Given a formula $\varphi$, 
we write $\rho \models \varphi$
to denote that $\varphi$ evaluates to true when its free variables are
assigned values from the environment $\rho$.

As discussed in Section~\ref{ssec:intro:approach},
one of our central ideas is to allow the names of the library procedures to
be referred to in the invariant; \eg, our running example becomes amenable
to our analysis using an invariant such as `(g = -1) $\vee$ (g = lastN *
\emph{factCache}(lastN-1))'.  We therefore allow the use of library
procedure names (in our simplified presentation, the name $\procname$) as free variables in
logical formulae. Correspondingly, 
we let each environment $\rho$ map each procedure name to a mathematical
function in addition to mapping variables to numeric values, and extend the
semantics of $\rho \models \varphi$ by substituting the values of both
variables and procedure names in $\varphi$ from the environment $\rho$.


Given an environment $\rho$, a procedure name $\procname$, and a mathematical function $f$, we will write
$\rho[\procname \mapsto f]$ to indicate the updated environment that maps
$\procname$ to the value $f$ and maps every
other variable $x$ to its original value $\rho[x]$.
We will write $(\rho,f) \models \varphi$ to denote that $\rho[\procname \mapsto f] \models \varphi$.

Given a state $\sigma = ((\stmt,\lmap)\gamma,\gmap)$, 
we define $\text{env}(\sigma)$ to be $\lmap \uplus \gmap$,  and given a
state $\sigma = ([],\gmap)$, we define $\text{env}(\sigma)$ to be just $\gmap$.
We write $(\sigma,f) \models \varphi$ to denote that $(\text{env}(\sigma),f) \models \varphi$.
For any execution $\pi$, we write $(\pi,f) \models \varphi$ if for every entry-state and exit-state
$\sigma$ in $\pi$, $(\sigma,f) \models \varphi$.
We now introduce another definition of observational purity.
\begin{definition}[Observational Purity wrt an Invariant]
\label{def:pureinv}
Given an invariant $\inv$, 
a library procedure $\proc$ is said to satisfy $\pureinv$ if
there exists a function $f$ such that for every trace $\pi$ of $\proc$,
$\outputval{\pi} = f(\inputval{\pi})$ and $(\pi,f) \models \inv$.
\end{definition}
It is easy to see that if procedure $\proc$ satisfies $\pureinv$ wrt any
given candidate invariant $\inv$, then $\proc$ is observationally pure as
per Definition~\ref{def:purity}. 




\newcommand{\existformula}{\psi^e}
\newcommand{\EA}{\textsc{ea}}
\newcommand{\IW}{\textsc{iw}}

\newcommand{\initformula}{\logicalformula{init}}

\section{Checking Purity Using a Theorem Prover}
\label{sec:vcgen}

In this section we provide two different approaches that, given a procedure
$\proc$ and a candidate invariant $\inv$, use a theorem prover to check
conservatively whether procedure $\proc$ satisfies $\pureinv$.

\subsection{Verification Condition Generation}

We first describe an adaptation of standard verification-condition generation
techniques that we use as a common first step in both our approaches.
Given a procedure $\procP$, a candidate invariant $\inv$, our goal is to compute a
pair $(\post,\vc)$ where $\post$ is a postcondition describing the state that exists after an execution of
$\text{body}(\procP)$ starting from a state that satisfies $\inv$, and $\vc$ is a verification-condition that must hold true
for the execution to satisfy its invariants and assertions.

We first transform the procedure body as below to create an internal representation that is input to the
postcondition and verification condition generator. In the internal representation, we allow the following
extra forms of statements (with their usual meaning): \code{havoc(x)}, \code{assume e}, and  \code{assert e}.
\begin{enumerate}
\item For any assignment statement ``\code{x := e}'' where \code{e} contains \code{x}, we introduce a new temporary
variable \code{t} and replace the assignment statement with ``\code{t := e; x := t}''.
\item For every procedure invocation ``\code{x := $\procname$(y)}'', we first ensure that \code{y} is a local variable (by introducing
a temporary if needed). We then replace the statement by the code fragment
``\code{assert $\inv$; havoc(g1); ... havoc(gN); assume $\inv \wedge$ x = $\procname$(y)}'',
where \code{g1} to \code{gN} are the global variables.

Note that the
function call has been eliminated, and replaced with an ``assume''
expression that refers to the function symbol $\procname$. In other words,
there are no function calls in the transformed procedure.
\item We replace the ``\code{return x}'' statement by ``\code{assert $\inv$}''.
\end{enumerate}
Let $\tbody(\proc, \inv)$ denote the transformed body of procedure $\proc$ obtained as above.


\begin{figure}
\[
\begin{array}{ll}
\postfn(\pre, \code{x := e}) &= (\exists \code{x}. \pre) \wedge (\code{x =
  e}) \ (\text{if } \code{x} \not\in \text{vars}(\code{e})) \\
\postfn(\pre, \code{havoc(x)}) &= \exists \code{x}. \pre \\
\postfn(\pre, \code{assume e}) &= \pre \wedge \code{e} \\
\postfn(\pre, \code{assert e}) &= \pre \\
\postfn(\pre, \stmtSA ; \stmtSB) &= \postfn( \postfn(\pre, \stmtSA), \stmtSB) \\
\multicolumn{2}{l}{\postfn(\pre, \code{if \expr{} then \stmtSA{} else \stmtSB{}}) = \postfn(\pre \wedge \expr, \stmtSA) \vee \postfn(\pre \wedge \neg \expr, \stmtSB)} \\
\\
\vcfn(\pre, \code{assert e}) &= (\pre \Rightarrow e) \\
\vcfn(\pre, \stmtSA ; \stmtSB) &= \vcfn(\pre, \stmtSA) \wedge \vcfn( \postfn(\pre, \stmtSA), \stmtSB) \\
\multicolumn{2}{l}{
\vcfn(\pre, \code{if \expr{} then \stmtSA{} else \stmtSB{}}) = \vcfn(\pre \wedge \expr, \stmtSA) \wedge \vcfn(\pre \wedge \neg \expr, \stmtSB)
} \\
\vcfn(\pre, \stmt) &= \text{true} (\text{for all other \stmt}) \\
\\
\mainvc{\procP}{\inv} = (\postfn(&\inv, \tbody(\proc,\inv)), \vcfn(\inv,
\tbody(\proc,\inv)) \wedge (\initstatefn(\proc) \Rightarrow \inv))
\end{array}
\]
\caption{Generation of verification-condition and postcondition.}
\label{fig:vcgen}
\end{figure}

We then compute postconditions as formally described in Fig.~\ref{fig:vcgen}.
This lets us compute for each program point $\ell$ in the procedure,
a condition $\varphi_{\ell}$ that describes what we expect to hold true when execution reaches $\ell$ if we start
executing the procedure in a state satisfying $\inv$ and if every recursive invocation of the procedure also
terminates in a state satisfying $\inv$. We compute this using the standard rules for the postcondition of a statement.
%
For an assignment statement ``\code{x := e}'', we use existential quantification over \code{x} to represent the value
of \code{x} prior to the execution of the statement. If we rename these existentially quantified variables with unique new
names, we can lift all the existential quantifiers to the outermost level. When transformed thus, the condition $\varphi_{\ell}$
takes the form $\exists x_1 \cdots x_n. \varphi$, where $\varphi$ is quantifier-free and $x_1, \cdots, x_n$ denote
intermediate values of variables along the execution path from procedure-entry to program point $\ell$.

We compute a verification condition $\vc$ that represents the conditions we must check to ensure that
an execution through the procedure satisfies its obligations: namely, that the invariant holds true at every call-site
and at procedure-exit. Let $\ell$ denote a call-site or the procedure-exit. We need to check that $\varphi_{\ell} \Rightarrow \inv$
holds. Thus, the generation verification condition essentially consists of the conjunction of this check over all call-sites
and procedure-exit.

Finally, the function $\mainvcfn$ computes the postcondition and verification condition for the entire procedure as shown
in Fig.~\ref{fig:vcgen}. Note that this adds the check that the initial state too must satisfy $\inv$ as the basis condition for
induction. $\initstatefn(\proc)$ is basically the formula  ``\code{g1 = c1}
$\wedge \ldots$ \code{gN = cN}'' (see Section~\ref{sec:background}).

\paragraph{Example}
We now illustrate the postcondition and verification condition generated from our factorial example
presented in Listing~\ref{lst:factorialSimple}. Listing~\ref{lst:factorialTransformed} shows the example
expressed in our language and transformed as described earlier (using function $\tbody$), using a
supplied candidate invariant $\inv$.

\begin{lstlisting}[float,language=c,mathescape=true,basicstyle=\scriptsize,caption= {Procedure factCache from
      Listing~\ref{lst:factorialSimple} transformed to incorporate a supplied candidate
      invariant $\inv$.}, label=lst:factorialTransformed]
g := -1;
lastN := 0;
factCache (n) {
  if(n <= 1) {
    result := 1;
  } else if (g != -1 && n == lastN) {
    result := g;
  } else {
    t1 := n-1;
    // t2 := factCache(t1);
    assert $\inv$;
    havoc (g); havoc (lastN);
    assume $\inv \wedge$ (t2 = factCache(t1));
    g := n * t2;
    lastN := n;
    result := g;
  }
  // return result;
  assert $\inv$;
}
\end{lstlisting}

Fig.~\ref{fig:pathCondition} illustrates the computation of postcondition and verification condition from
this transformed example. In this figure, we use $\cspre$ to denote the precondition computed to hold
just before the recursive callsite, and $\cspost$ to denote the postcondition computed to hold just
after the recursive callsite. The postcondition $\post$ (at the end of the procedure body) is itself
a disjunction of three path-conditions representing execution through the three different paths in
the program. In this illustration, we have simplified the logical conditions by omitting useless existential
quantifications (that is, any quantification of the form $\exists x. \psi$ where $x$ does not occur in $\psi$).
Note that the existentially quantified \code{g} and \code{lastN} in $\cspost$ denote the values of these
globals \empty{before} the recursive call. Similarly, the existentially quantified \code{g} and \code{lastN} in
$\path_3$ denote the values of these globals when the recursive call terminates, while the free variables
\code{g} and \code{lastN} denote the final values of these globals.

\begin{figure}
\begin{align*}
\initstatefn(\proc) &= \code{(g = -1)} \wedge \code{(lastN = 0)} \\
\path_1 &= \inv \wedge (\code{n <= 1}) \wedge (\code{result = 1}) \\
\path_2 &= \inv \wedge \neg(\code{n <= 1}) \wedge (\code{g != 1}) \wedge (\code{n = lastN}) \wedge (\code{result = g}) \\
\cspre &= \inv \wedge \neg(\code{n <= 1}) \wedge \neg((\code{n = lastN}) \wedge (\code{result = g})) \wedge (\code{t1 = n-1}) \\
\cspost &= (\exists \code{g} \exists \code{lastN} \; \cspre) \wedge \inv \wedge (\code{t2 = \emph{factCache}(t1)}) \\
\path_3 &= (\exists \code{g} \exists \code{lastN} \; \cspost) \wedge (\code{g = n * t2}) \wedge (\code{last N = n}) \wedge (\code{result = g}) \\
\post &= \path_1 \vee \path_2 \vee \path_3 \\
\vc &= (\cspre \Rightarrow \inv) \wedge (\post \Rightarrow \inv) \wedge (\initstatefn(\proc) \Rightarrow \inv)
\end{align*}
\caption{
The different formulae computed from the procedure in Listing~\ref{lst:factorialTransformed} 
by our postcondition and verification-condition computation.
}
\label{fig:pathCondition}
\end{figure}


\subsection{Approach 1: Existential Approach}

Let $\proc$ be a procedure with input parameter $n$ and return variable $r$.
Let\\ $\mainvc{\proc}{\inv}$ = $(\post,\vc)$.
Let $\existformula$ denote the formula $\vc \wedge (\post \Rightarrow (r = p(n)))$.
Let $\overline{x}$ denote the sequence of all free variables in $\existformula$ except for $p$.
We define $\EA(\proc,\inv)$ to be the formula $ \forall \overline{x}. \existformula$.

In this approach, we use a theorem prover to check whether $\EA(\proc,\inv)$ is satisfiable.
As shown by the following theorem, satisfiability of $\EA(\proc,\inv)$ establishes that $\proc$
satisfies $\pureinv$.

\begin{theorem}
\label{theorem:EA}
A procedure $\proc$ satisfies $\pureinv$ if
$\exists p. \EA(\proc,\inv)$ is a tautology
(which holds iff $\EA(\proc,\inv)$ is satisfiable).
\end{theorem}

\begin{proof}
Note that $p$ is the only free variable in $\EA(\proc,\inv)$. Assume that $[p \mapsto f]$ is a
satisfying assignment for $\EA(\proc,\inv)$.
We prove that for every trace $\pi$ the following hold:
(a) $\outputval{\pi} = f(\inputval{\pi})$ and
(b) If $(\initial{\pi},f)$ satisfies $\inv$, then $(\final{\pi},f)$ also satisfies $\inv$.

The proof is by contradiction. Let $\pi$ be the shortest trace that does not satisfy at least
one of the two conditions (a) and (b).

Let us first consider a trace $\pi$ without any sub-traces (\ie, without a procedure call).
Consider any transition $\sigma_i \sssemP \sigma_{i+1}$ in $\pi$
caused by an assignment statement $\stmt$. Our verification-condition generation produces
a precondition $\pre_{\stmt}$ and a postcondition $\post_{\stmt}$ for $\stmt$. Further, this generation
guarantees that if $(\sigma_i,f)$ satisfies $\pre_{\stmt}$, then $(\sigma_{i+1},f)$ satisfies $\post_{\stmt}$.
This lets us prove (via induction over $i$), that if $(\initial{\pi},f)$ satisfies $\inv$,
then $(\sigma_{i+1},f)$ satisfies $\post_{\stmt}$.

We now consider a trace $\pi$ that contains a sub-trace $\pi' = \sigma_i \cdots \sigma_j$ corresponding to a procedure call
statement $\stmt$. Our verification-condition generation produces a precondition $\pre_{\stmt}$,
a postcondition $\post_{\stmt}$  and a conjunct $\pre_{\stmt} \Rightarrow \inv$ in $\vc$ for $\stmt$.
We extend the induction over $i$ to handle recursive calls as below.
Our inductive hypothesis guarantees that $(\sigma_i,f)$ satisfies $\pre_{\stmt}$.
Further, we know that $[p \mapsto f]$ is a satisfying assignment for $\EA(\proc,\inv)$,
which includes the conjunct $\pre_{\stmt} \Rightarrow \inv$. Hence, it follows that  $(\sigma_i,f)$ satisfies $\inv$.
Since $\pi'$ is a shorter trace than $\pi$, we can assume that it satisfies conditions (a)
and (b). This is sufficient to guarantee that $(\sigma_j,f)$ satisfies $\post_{\stmt}$.

Thus, we can establish $(\final{\pi},f)$ satisfies $\post$ and $\inv$. Further, since 
$\EA(\proc,\inv)$ includes the conjunct $\post \Rightarrow (r = p(n))$, it follows that
$\outputval{\pi} = f(\inputval{\pi})$.
\end{proof}

\subsection{Approach 2: Impurity Witness Approach}

The existential approach presented in the previous section has a drawback. Checking satisfiability of $\EA(\proc,\inv)$
is hard because it contains universal quantifiers and existing theorem provers do not work well enough for this
approach. We now present an approximation of the existential approach that is easier to use with existing theorem
provers. This new approach, which we will refer to as the impurity witness approach, reduces the problem to
that of checking whether a quantifier-free formula is unsatisfiable, which is better suited to the capabilities of
state-of-the-art theorem provers. This approach focuses on finding a counterexample to show that the
procedure is impure or it violates the candidate invariant.

Let $\proc$ be a procedure with input parameter $n$ and return variable $r$.
Let $\mainvc{\proc}{\inv}$ = $(\post,\vc)$.
Let $\post_\alpha$ denote the formula obtained by replacing every free variable $x$ other than $p$ in $\post$
by a new free variable $x_\alpha$. Define $\post_\beta$ similarly.
Define $\IW(\proc, \inv)$ to be the formula $(\neg \vc) \vee (\post_\alpha \wedge \post_\beta \wedge (n_\alpha = n_\beta) \wedge (r_\alpha \neq r_\beta))$.

The impurity witness approach checks whether $\IW(\proc, \inv)$ is satisfiable. This can be done by separately checking
whether $\neg \vc$ is satisfiable and whether $(\post_\alpha \wedge \post_\beta \wedge (n_\alpha = n_\beta) \wedge (r_\alpha \neq r_\beta))$
is satisfiable. As formally defined, $\vc$ and $\post$ contain embedded existential quantifications. As explained earlier,
these existential quantifiers can be moved to the outside after variable renaming and can be omitted for a satisfiability check.
(A formula of the form $\exists \overline{x}. \psi$ is satisfiable iff $\psi$ is satisfiable.)
As usual, these existential quantifiers refer to intermediate values of variables along an execution path.
Finding a satisfying assignment to these variables essentially identifies a possible execution path (that
satisfies some other property).

\begin{theorem}
A procedure $\proc$ satisfies  $\pureinv$ if $\IW(\proc, \inv)$ is unsatisfiable.
\end{theorem}

\begin{proof}
We prove the contrapositive.

We say that a pair of feasible executions $(\pi_1, \pi_2)$ is an impurity witness if there is a trace
$\pi_a$ in $\pi_1$ and a trace $\pi_b$ in $\pi_2$ such that $\pi_a$ and $\pi_b$ have the same input
value but different return values. Otherwise, we say that $(\pi_1, \pi_2)$ is pure. We extend
this notion and say that a single execution $\pi$ is pure if $(\pi,\pi)$ is pure.

We say that a function $f$ is \emph{compatible} with a set of executions $\Pi$ if for every trace
$\pi \in \Pi$, $\outputval{\pi} = f(\inputval{\pi})$.

We say that a pure feasible execution $\pi$ is a $\inv$-violation witness if there is some function $f$
that is compatible with $\{\pi\}$ such that $(\pi,f) \nvDash \inv$. Otherwise, we say that $\pi$ satisfies
the invariant. Note that this definition introduces a conservative approximation as we explain later.

We will consider \emph{minimal} witnesses of the following form.
We say that an impurity witness $(\pi_1, \pi_2)$ is minimal if for every proper prefix $\pi_1'$ of
$\pi_1$ and every proper prefix $\pi_2'$ of $\pi_2$, the following hold:
(a) $(\pi_1',\pi_2)$ is pure,
(b) $(\pi_1,\pi_2')$ is pure,
(c) $\pi_1'$ satisfies the invariant, and
(d) $\pi_2'$ satisfies the invariant.

We say that an $\inv$-violation witness $\pi$ is minimal if no proper prefix $\pi'$ of
$\pi$ is a $\inv$-violation witness.

If $\proc$ does not satisfy $\pureinv$, then there exists a minimal impurity witness
or a minimal $\inv$-violation witness. Note that our definition of $\inv$-violation witness is
a conservative approximation and the converse of the preceding claim does not hold.
A  $\inv$-violation witness does not mean that $\proc$ does not satisfy $\pureinv$.

In the first case, $(\post_\alpha \wedge \post_\beta \wedge (n_\alpha = n_\beta) \wedge (r_\alpha \neq r_\beta))$
must be satisfiable (as we show below). In the second case, $\neg \vc$ must be satisfiable (as we show below).
Thus,  $\IW(\proc, \inv)$  is satisfiable in either case. The theorem follows.

We establish the above result as below.
If we have a trace that satisfies the invariant and the set of its sub-traces are pure,
then the valuations assigned to variables by the trace satisfies $\post$.
Thus, if $(\pi_1,\pi_2)$ are a minimal impurity witness, let $\pi_a$ and $\pi_b$
be the two traces in $\pi_1$ and $\pi_2$ that are incompatible.
We can assign values to variables in $\post_\alpha$ from $\pi_a$, and
assign values to variables in $\post_\beta$ from $\pi_b$ to get a satisfying
assignment for $(\post_\alpha \wedge \post_\beta \wedge (n_\alpha = n_\beta) \wedge (r_\alpha \neq r_\beta))$.

If $\pi$ is a minimal $\inv$-violation witness, let $\pi'$ be the trace that contains the invariant violation.
Assigning values to variables as in $\pi'$ produces a satisfying assignment for
$\neg \vc$.
\end{proof}

\section{Generating the Invariant}
\label{sec:invariant}

We now describe a simple but reasonably effective semi-algorithm for
generating a candidate invariant automatically from the given
procedure. Our approach of Section~\ref{sec:vcgen} can be used with a manually
provided invariant or the candidate invariant generated by this
semi-algorithm (whenever it terminates).

The invariant-generation approach is iterative and computes a sequence of progressively weaker
candidate invariants $I_0, I_1, \cdots$ and terminates if and when $I_m \equiv I_{m+1}$, at
which point $I_m$ is returned as the candidate invariant.
The initial candidate invariant $I_0$ captures the initial values of the global variable.
In  iteration $k$, we apply a procedure similar to the one described in Section~\ref{sec:vcgen} and
compute the strongest conditions that hold true at every program point if the execution of the
procedure starts in a state satisfying $I_{k-1}$ and if every recursive invocation terminates in a
state satisfying $I_{k-1}$. We then take the disjunction of the conditions computed at the points before the
recursive call-sites and at the end of the procedure, and existentially quantify all local variables.
We refer to the resulting formula as $\nextfn(I_{k-1}, \tbody(\proc,I_{k-1}))$.
We take the disjunction of this formula with $I_{k-1}$ and simplify it to get $I_k$.

%

In the following formalization of this semi-algorithm, we exploit the fact that the \code{assert}
statements are added precisely at every recursive callsite and end of procedure and
these are the places where we take the conditions to be disjuncted.
\begin{figure}
\[
\begin{array}{ll}
\multicolumn{2}{l}{I_0 = \initstatefn(\proc)} \\
\multicolumn{2}{l}{I_{k} = \simplifyfn(I_{k-1} \vee \nextfn(I_{k-1}, \tbody(\proc,I_{k-1})))} \\
\\
\nextfn(\pre, \code{assert e}) &= \exists \ell_1 \cdots \ell_m \pre (\text{where $\ell_1, \cdots, \ell_m$ are local variables in $\pre$})\\
\nextfn(\pre, \stmtSA ; \stmtSB) &= \nextfn(\pre, \stmtSA) \vee \nextfn( \postfn(\pre, \stmtSA), \stmtSB) \\
\multicolumn{2}{l}{
\nextfn(\pre, \code{if \expr{} then \stmtSA{} else \stmtSB{}}) = \nextfn(\pre \wedge \expr, \stmtSA) \vee \nextfn(\pre \wedge \neg \expr, \stmtSB)
} \\
\nextfn(\pre, \stmt) &= \text{false} (\text{for all other \stmt})
\end{array}
\]
\caption{Iterative computation of invariant.}
\label{fig:invgen}
\end{figure}

In our running example, $I_0$  is`g = -1 $\wedge$ lastN = 0'.
Applying $\nextfn$ to $I_0$
yields $I_0$ itself as the pre-condition at the
point just before the recursive call-site, and `(g = -1 $\wedge$ lastN = 0) $\vee$ g = lastN *
$\procname$(lastN-1)' (after certain simplifications) as the pre-condition
at the end of the
procedure. Therefore, $I_1$ is `(g = -1 $\wedge$ lastN = 0) $\vee$ g = lastN *
$\procname$(lastN-1)'. When we apply $\nextfn$ to $I_1$,
the computed pre-conditions are $I_1$ itself at both the program points
mentioned above. Therefore, the approach terminates with $I_1$ as the
candidate invariant.


\section{Evaluation}\label{sec:experiments}

We have implemented our OP checking approach as a prototype using the Boogie
framework~\cite{leino2008boogie}, and have evaluated the approach using
this implementation on several examples. The objective of this evaluation
was primarily a sanity check, to test how our approach does on a set of
OP as well as non-OP procedures.

We tried several simple non-OP programs, and our implementation terminated
with a ``no'' answer on all of them.  We also tried the approach on several
OP procedures: (1) the `factCache' running example, (2) a version of
a factorial procedure that caches all arguments seen so far and their
corresponding return values in an array, (3) a version of factorial that caches
only the return value for argument value 19 in a scalar variable, (4) a
recursive procedure that returns the $n^\mathit{th}$ Fibonacci number and
caches all its arguments and corresponding return values seen so far in an
array, and (5) a ``matrix chain multiplication'' (MCM) procedure.
The last example  is based
on dynamic programming, and hence naturally uses a table to memoize 
results for sub-problems. Here, observational purity implies that the procedure always
returns the same solution for a given sub-problem, whether a hit was found
in the table or not.  The appendix of a technical report associated with
this paper  depicts all the procedures
mentioned above as created by us directly in Boogie's language, as well as
the invariants that we supplied manually (in  SMT2
format).

It is notable that our ``existential approach'' causes the theorem prover
to not scale to even simple examples. The ``impurity witness'' approach
terminated on all the examples mentioned above with a ``yes'' answer,
with the theorem prover
taking less than 1 second on each example.

\section{Related Work}\label{sec:related}

The previous work that is most closely related to our work is by Barnett et
al.~\cite{barnett200499,barnett2006allowing}. Their approach is based on
the same notion of observational purity as our approach. Their approach is
structurally similar to ours, in terms of needing an invariant, and using
an inductive check for both the validity of the invariant as well as the
uniqueness of return values for a given argument.  However, their approach
is based on a more complex notion of invariant than our approach, which
relates pairs of global states, and does not use a function symbol to
represent recursive calls within the procedure. Hence, their approach does
not extend readily to recursive procedures; they in fact state that ``there
is a circularity - it would take a delicate argument, and additional
conditions, to avoid unsoundness in this case''. Our idea of allowing the
function symbol in the invariant to represent the recursive call allows
recursive procedures to be checked, and also simplifies the specification
of the invariant in many cases. 

Cok et al.~\cite{cok2008extensions}  generalize the work of Barnett
et al.'s work, and suggest classifying procedures into categories ``pure'',
``secret'', and ``query''. The ``query'' procdures are observationally
pure. Again, recursive procedures are not addressed.

Naumann~\cite{naumann2007observational} proposes a notion of observational
purity that is also the same as ours. Their paper gives a rigorous but manual
methodology for proving the observational purity of a
given procedure. Their methodology is not similar to ours; rather, it is
based finding a \emph{weakly pure} procedure that simulates the given
procedure as far as externally visible state changes and the return value
are concerned. They have no notion of an invariant that uses a function
symbol that represents the procedure, and they don't explicitly address the
checking of recursive procedures.

There exists a significant body of work on identifying differences between
two similar procedures.  For instance, differential assertion
checking~\cite{lahiri2013differential} is a representative from this body,
and is for checking if two procedures can ever start from the same state
but end in different states such that exactly one of the ending states
fails a given assertion. Their
approach is based on logical reasoning, and accommodates recursive
procedures. Our impurity witness approach has some similarity with their
approach, because it is based on comparing the given procedure with
itself. However, our comparison is stricter, because in our setting,
starting with a common argument value but from different global states that
are both within the invariant should not cause a difference in the return
value. Furthermore, technically our approach is different because we use an
invariant that refers to a function symbol that represents the procedure
being checked, which is not a feature of their invariants. Partush et
al.~\cite{partush2013abstract} solve a similar problem as differential
assertion checking, but using abstract interpretation instead of logical
reasoning.

There is a substantial body of work on checking if a procedure is
\emph{pure}, in the sense that it does not
modify any objects that existed before the procedure was invoked, and does not
modify any  global variables. Salcianu et
al.~\cite{sualcianu2005purity} describe a static analysis to check purity.
Various tools exist, such as
JML~\cite{leavens2008jml} and Spec\#~\cite{barnett2004spec}, that use logical
techniques based on annotations to prove procedures as pure.  Purity is a
more restrictive notion than observational purity; procedures such
as our `factCache' example are observationally pure, but not pure because
they use as well as update state that persists between calls to the
procedure.

%

%
%
\bibliographystyle{splncs}
\bibliography{references}



\appendix
\begin{subappendices}
\renewcommand{\thesection}{\Alph{section}}
\section{Proofs}

\begin{lemma}
If there exists no impurity witness $(\pi_1,\pi_2)$ and there exists
no $\inv$-violation witness $\pi$, then $\proc$ satisfies $\pureinv$.
\end{lemma}

\begin{proof}
Consider the set $\Pi$ of all feasible executions of $\proc$.
If there exists no impurity witness, then there exists a function $f$ compatible with $\Pi$.
(We take the partial function $\{ (\inputval{\pi}, \outputval{\pi}) \; | \; \pi \in \Pi \}$ and
extend it to be a total function.)
Further, since there exists no $\inv$-violation witness, we are guaranteed that $(\pi,f) \models \inv$
for every $\pi \in \Pi$.
\end{proof}

\begin{lemma}
If there exists a $\inv$-violation witness $\pi$, then there exists a minimal $\inv$-violation witness.
\end{lemma}

\begin{lemma}
If there exists an impurity witness, then there exists a minimal impurity witness or a minimal
$\inv$-violation witness.
\end{lemma}

\section{Examples}

\subsection{Factorial, caching single result}
Invariant:
\begin{verbatim}
(or (= nineteen (- 0 1)) (= nineteen (* (FactSingle 18) 19)))
\end{verbatim}

\begin{lstlisting}[language=c, caption= {Returns factorial of `n',
      and memoizes result for argument value
      `19' in global variable `nineteen'.}, label=lst:fact19]
var nineteen: int;
/* invariant: nineteen = -1 ||
  nineteen = 19*FactSingle(18)   */
procedure {:entrypoint} FactSingle(n: int) returns (r: int) modifies nineteen;{
  if( n <= 1) { r := 1;}
  else {
    if( n == 19) {
      if( nineteen == -1) {
        call nineteen := FactSingle(18);
        nineteen := nineteen * 19;
        r := nineteen;
      } else {
        r := nineteen;
      }
    } else {
      call r := FactSingle( n - 1);
      r := n * r;
    }
  }
}
\end{lstlisting}

\subsection{Factorial, caching in  array}
Invariant:
\begin{verbatim}
(forall ((k Int) ) (or (= (select g k) 0) (= (*
k (FactArray (- k 1))) (select g k))))
\end{verbatim}
\begin{lstlisting}[language=c, caption= {Returns factorial of `n', 
      and caches results for all argument values greater than 1 in global
    array `g'},
    label=lst:factArrayImpl]
var g: [int] int;
/* invariant: forall k. g[k] = 0 ||
  g[k] = k*FactArray(k-1)   */  
procedure {:entrypoint} FactArray(n: int) returns (r: int) modifies g;{
  var k :int;
  if( n <= 1) { r := 1;}
  else {
    if( g[n] == 0) {
      call k := FactArray(n - 1);
      g[n] := k * n;
    } 
    r  := g[n];
  }
}
\end{lstlisting}

\subsection{Factorial, caching only last-seen argument}
Invariant:
\begin{verbatim}
(or (= g (- 0 1)) (= g (* (FactRecent (+ (- 0 1) lastN)) lastN)))
\end{verbatim}
\begin{lstlisting}[language=c, caption= {Returns factorial of `n', 
      and caches last seen argument in lastN and corresponding return value
    in g}, label=lst:factorialRecent]
var lastN: int;
var g: int;
//invariant : g = -1 || g = lastN * FactRecent(lastN - 1) 
procedure {:entrypoint} FactRecent(n: int) returns (r: int) modifies lastN, g;{
  if( n <= 1) { r := 1;}
  else {
     if( n == lastN && g != -1) {
        r := g;
      } else {
      call r := FactRecent( n - 1);    
      r := n * r;
      lastN := a;
      g := r;
    }
  }
}
\end{lstlisting}

\subsection{Fibonacci, caching in array}
Invariant:
\begin{verbatim}
(forall (k Int) (or (= (select cache k) 0) (= (+ (fib (- k 1)) (fib (- k 2)))
  (select cache k))))
\end{verbatim}
\begin{lstlisting}[language=c, caption= {Returns $n$th Fibonacci number, and
      caches all arguments and return values seen so far in global array `cache'}]
/* invariant : forall k. cache[k] = 0 OR cache[k] = fib(k -1) + fib( k -2) */
var cache:[int] int;
procedure {:entrypoint} fib(n: int) returns (r: int) modifies cache;{
  var a, b : int;
  if( n <= 2) {
    r := 1;
  } else {
    if(cache[n] != 0) {
      r := cache[n];
    } else {
      call a := fib(n -1);
      call b := fib(n -2);
      r := a + b;
      cache[n] := r;
    }
  }
}
\end{lstlisting}

\subsection{Matrix Chain Multiplication}
Invariant:
\begin{verbatim}
 (forall ((i Int) (j Int)) (or (= (select m i j) (- 0 1))
   (= (chooseSplit i j i (- 0 1)) (select m i j))))
\end{verbatim}

\begin{lstlisting}[language=c, caption= {Returns the minimum number of multipications needed to multiply a sequence of matrices.}, label=lst:mcm]
var p: [int] int;
var m: [int, int] int;
/* invariant : forall i, j. m[i, j] = -1 OR m[i, j] = chooseSplit(i, j, i, -1) */
procedure {:entrypoint} mcm(i: int, j: int) returns (r: int) modifies m;{
        var k, q : int;
        var a, b :int;
        if(i == j) {
                m[i, j] := 0;
                r := 0;
        } else {
                if( m[i, j] > 0) {
                        r := m[i, j];
                } else {
                        k := i;
                        call r := chooseSplit(i, j, k, m[i, j]);
                        m[i, j] := r;
                }
        }
}

procedure chooseSplit(i: int, j: int, k:int, min :int) returns (r: int) modifies m;{
        var a, b, q : int;
        var min1 :int;
        if(k >= j) {
                r:= min;
        } else {        
                call a := mcm(i, k);
                call b := mcm(k+1, j);
                q := a + b + p[i-1] * p[k] * p[j];
                if( q < min) {
                  min1 := q;
                } else {
                  min1 := min;
                }
                call r := chooseSplit(i, j, k + 1, min1);
        }       
}
\end{lstlisting}

\subsubsection{Notes about how this program was verified.}
In this example, `p' is an input array, and stores the dimensions of the matrices
to be multiplied. p$[i]$ and p$[i+1]$ store the dimensions of the $i^\mathit{th}$
matrix to be multiplied. Table `m' is computed by the procedure. m$[i,j]$
stores the minimum number of scalar multiplications required
to multiply the sequence of matrices from the $i^\mathit{th}$ matrix
through the $j^\mathit{th}$ matrix. The procedure `mcm'
basically returns m[i,j] given arguments i and j. chooseSplit is a
tail-recursive procedure, which is used to consider all possible k, i $<$ k
$\leq$ j, at which to split the sequence $i, i+1, \ldots, j$. This search
is required to find the optimal split point.

Note, mcm and chooseSplit call each other. Although we presented our
approach as if it can check a given single procedure, it can also check if
a set of mutually-recursive procedures are all OP. The idea is that each
procedure needs an invariant, and this invariant can use function symbols
to refer to all procedures called by this procedure. Each procedure can
then be checked independently using our approach as-is. Intuitively, this
works because we check all the procedures, and while checking any procedure
we assume as an inductive hypothesis that the other procedures are OP.

The invariant shown in the example above is for procedure `mcm'. We applied
our tool only on procedure mcm. We did not apply our tool on procedure
`chooseSplit', because currently our tool is not set up to analyze all
procedures in a multi-procedure program. However, it is obvious that the
impurity witness approach would certify procedure chooseSplit as OP, even
without any invariant. This is because this procedure does not refer to or
update the mutable global variable `m'. Therefore, for any tuple of
argument values, the procedure would always follow the same path no matter
how many times this tuple is repeated. Therefore, if the calls encountered in
this path are assumed to be to observationally pure procedures (inductive
hypothesis), then the path must return the same value every time.

\end{subappendices}

\end{document}